\documentclass{article}

\usepackage{orcidlink}
\usepackage{authblk}
\usepackage{natbib}

\usepackage{moreverb,url}

\usepackage{graphicx} 
\usepackage{amsmath,amsthm,amssymb} 
\usepackage{xspace} 
\usepackage{cleveref}
\usepackage{algorithm2e}
\usepackage{tikz}
\usetikzlibrary{matrix,calc}
\usepackage{mathabx} 
\usepackage{enumitem}
\usepackage{booktabs}
\usepackage{multirow}
\usepackage{longtable}
\usepackage[a4paper, left=2.5cm, right=2.5cm]{geometry}

\Crefname{algocf}{Algorithm}{Algorithms}
\RestyleAlgo{ruled}
\SetKwInOut{Input}{input}
\SetKwInOut{Output}{output}
\newtheorem{definition}{Definition}
\newtheorem{problem}{Problem}
\newtheorem{observation}{Observation}
\crefname{observation}{observation}{observations}
\Crefname{observation}{Observation}{Observations}
\crefname{problem}{problem}{problems}
\Crefname{problem}{Problem}{Problems}

\newtheorem{theorem}{Theorem}
\newtheorem{lemma}{Lemma}

\newcommand\msaij[2]{\mathsf{MSA}[#1,#2]}
\newcommand\gap{\mathtt{-}}

\DeclareMathOperator{\maxheight}{H}

\DeclareMathOperator{\segmentheight}{H}

\DeclareMathOperator{\spell}{spell}
\DeclareMathOperator{\stringsize}{m}
\DeclareMathOperator{\size}{N}
\DeclareMathOperator{\gapsize}{N_\varepsilon}
\newcommand{\trivialsegmentationvertical}{S^{\,\rotatebox[origin=c]{90}{$\equiv$}}}
\newcommand{\trivialsegmentationhorizontal}{S^{\equiv}}
\newcommand{\minUcard}{$\mathsf{min}$-$U$-$\mathsf{cardinality}$}
\newcommand{\minLsize}{$\mathsf{min}$-$L$-$\mathsf{size}$}

\newcommand{\preclex}{\mathbin{\prec_{\mathrm{lex}}}}
\newcommand{\preceqcolex}{\mathbin{\preceq_{\mathrm{colex}}}}
\newcommand{\preccolex}{\mathbin{\prec_{\mathrm{colex}}}}
\usepackage{soul,xcolor}

\begin{document}

\title{Pangenome Optimization via Elastic Degenerate Strings}

\author[1]{Nicola Rizzo}
\author[2]{Sebastian Visan-Draghicescu}
\author[3]{Nadia Pisanti}
\author[4]{Veli M{\"a}kinen}

\affil[1]{Department of Computer Science, University of Helsinki, Finland, \texttt{nicola.rizzo@helsinki.fi} \orcidlink{0000-0002-2035-6309}}
\affil[2]{Department of Computer Science, University of Helsinki, Finland, \texttt{Sebastian.Visan-Draghicescu@helsinki.fi}}
\affil[3]{Department of Computer Science, University of Pisa, Italy, \protect\\
\texttt{nadia.pisanti@unipi.it} \orcidlink{0000-0003-3915-7665}}
\affil[4]{Department of Computer Science, University of Helsinki, Finland, \texttt{veli.makinen@helsinki.fi} \orcidlink{0000-0003-4454-1493}}

\date{}

\maketitle

\begin{abstract}
An \emph{Elastic Degenerate String} (EDS, or ED-string) is a sequence of string sets. A pangenome, consisting of variations observed in a population along the genome sequences, can be naturally encoded as an EDS. Pattern matching and comparison problems on pangenome representations such as EDSes have been widely studied in the literature, but optimizing the pangenome properties during its construction has been largely omitted. We fill this gap by showing how methods originally developed for the related problem of founder reconstruction can be adapted to minimize, in linear time, the total cardinality of the EDS sets or the total size of the EDS strings, given suitable multiple alignments representing the input data. We provide an implementation for the minimum-cardinality criterion in a tool \texttt{mincard}, and conduct the first experiments on scalable pangenome optimization via EDSes. The code and experiments are available at \url{https://github.com/algbio/eds}.
\end{abstract}

\section{Introduction}

An \emph{Elastic Degenerate String} (EDS, or ED-string) is a sequence of string sets representing all possible ordered combinations. Given its highly recombinant nature, an EDS can be used to represent the variations observed in the genomes of a population, making it a possible \emph{pangenome} representation~\citep{Maretal16,DBLP:journals/nc/BaaijensBBVPRS22}. Matching and comparison problems on EDSes have gained attention as intriguing combinatorial problems (see Section~\ref{sec:related}), but also as practical abstractions to enhance read mapping workflows with pangenomic information \citep{BOO23,DBLP:journals/bioinformatics/0001CM25}. The algorithmic results and running times of tools that deploy EDSes depend on the properties of EDSes, but a rigorous study of how to optimize these properties has been missing.
We fill this gap by showing how techniques originally developed for related problems on founder reconstruction can be adapted to minimize, in linear time, two fundamental EDS properties, as defined in the following formulations through a segmentation framework. 

Given a \emph{multiple sequence alignment (MSA)} representing the input, defining an EDS can basically boil down to partitioning the MSA into the segments that represent the sets of the EDS:  In an MSA, the $r$ input sequences are made to be of equal length $c$ by adding gaps `\texttt{-}' forming a matrix of $r$ rows and $c$ columns which we will denote $\msaij{1..r}{1..c}$. Computing an optimal partition of an interval $[1..c]$, that is, a \emph{segmentation}, is thus a problem that arises naturally when dealing with the transformation of sequences that are aligned into $c$ positions, according to 
the constraints and an optimality measure imposed on the solution. Figure~\ref{fig:example} gives an example of such segmentation on an MSA of $6$ sequences. 
While MSA construction is NP-hard \citep{Mai78}, practical heuristics have been developed to obtain high quality MSAs \citep{Notredame07}, and recent efforts have concentrated into scaling MSA construction to human-chromosome scale~\citep{olbrich2025generating}.
In this paper, we assume an MSA to be given as input, and we look for rigorous segment-based ways to convert the MSA into a sequence of string sets, that is, into an EDS.
As can be seen in Figure~\ref{fig:example}, an MSA segmentation directly induces an EDS by interpreting the strings of each segment to be part of a set.

\begin{figure}[htp]
    \centering
    \includegraphics{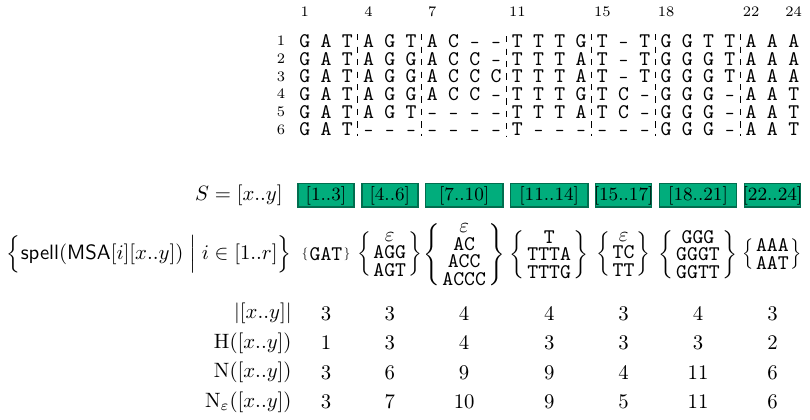}
\caption{Example of an $\msaij{1..6}{1..24}$ and a segmentation $S$ containing $k = 7$ segments. The resulting EDS (which we denote with the same symbols $S$ as the segmentation) is the sequence of sets shown in curly brackets.
The height $\maxheight(S)$ of $S$ is the maximum size of its sets, and hence $\maxheight(S)=4$; the cardinality $\stringsize(S)$ is 19 being the total number of elements in the sets; the size (the total number of letters ignoring the gap $\varepsilon$) $\size(S)$ is 48, and the gap-aware size $\gapsize(S)$ is 51.}\label{fig:example}
\end{figure}

The transformation from an MSA to an EDS through segmentation can be seen as a lossy compression: by collapsing equal strings inside each segment, we ignore the exact genomic paths. This can be a desired feature, for example, for hiding the sensitive full input sequence by publishing only the EDS (see \cite{Blindenbach2026.02.16.706152} for a recent privacy approach for pangenome graphs). A more rudimentary reason for such conversion is that EDSes and related pangenome graphs do compress well and enable fast processing and analysis of variation in populations \citep{Maretal16}.
A notable feature of this scheme is that if a query string occurs in a row of the MSA, then it will also occur in the EDS, but not the other way around; therefore, an EDS can create artifacts with respect to the original MSA. Various ways to optimize the transformation to minimize the amount of artifacts while not sacrificing privacy could be considered. Moreover, minimizing a specific property is useful when the downstream analyses of the resulting EDS employ algorithms whose complexity depends from that property. In this paper, we specifically target \emph{cardinality} and \emph{size}. The latter is the total number of characters and appears in the time and space complexity of basically any algorithm that has been designed to analyze an ED string. The former one is the total number of sequences: whenever the algorithmic task involves two EDSes (e.g.\ in comparison and distance computation algorithms), non-trivial methods avoid the quadratic size cost and rather depend on the cardinality  (see for example \cite{Ascetal24}).

We initiate the study on how to efficiently construct an EDS by looking at the selected optimization criteria mentioned above and defined in~\Cref{sec:defs}. For these parameters we give
linear-time algorithms (\Cref{sec:algs}) to minimize the total cardinality of the EDS sets (\emph{cardinality}) and the total size of the strings in the EDS sets (\emph{size}) under specific assumptions: the algorithms directly work for gapless MSAs, and they serve as heuristics when gaps are interpreted as part of the alphabet.
In the latter case, the optimization criteria should be adjusted accordingly: see \Cref{sec:experiments}  for further discussion and experiments on general MSAs. We start by reviewing the related work in \Cref{sec:related}.
A preliminary and partial version of this paper appeared as part of a Festschrift \citep{DBLP:conf/birthday/0001MP25}. 

\section{Related work\label{sec:related}}
We present the main results on the segmentations and partitions of aligned sequences.
\begin{description}
    \item[Founder reconstruction.]   
    \cite{Ukkonen02founder} studied the problem of explaining a given set of $r$ aligned sequences representing $c$ haplotype sites (with no insertions or deletions) with the recombination of a few \emph{founder sequences}, under the simple \emph{crossover} model of recombination (sequence switching at any of the $c$ aligned positions); minimizing the founder set is trivial if the recombination can happen at any position, but if the recombinations are allowed only at restricted positions to be chosen, then the resulting segmentation problems---minimizing the maximum cardinality of a segment (the number of founders in \citep{Ukkonen02founder}) given a lower bound on the segment length, or given a target average segment length size to surpass---can be solved in polynomial time.

    \cite{DBLP:conf/wabi/NorriCKM18,Norri19scalablefounder} solved one of the problems by Ukkonen, the segmentation of aligned sequences minimizing the maximum height given a segment-length lower bound $L$, in linear $O(rc)$ time and $O(r + L)$ space.
    The solution proceeds in a dynamic programming fashion, from left to right, and for each aligned position $y$ it uses the \emph{positional Burrows--Wheeler Transform}~\citep{Dur14}, enhanced by a few other arrays, to efficiently compute all the possible heights of a segment ending at $y$ and their optimal recursive values.

    \cite{Cazaux19columnstream} studied the related problem of maximizing the minimum segment length or the average segment length of any segmentation that has height of at most a given $H$, motivated by minimizing the size of the resulting pangenomic index based on founder sequences.
    To do so, they remark that the approach by \cite{DBLP:conf/wabi/NorriCKM18,Norri19scalablefounder} fits into a more general left-to-right \emph{column-stream} model of computation, and use \emph{range maximum queries} to efficiently compute the recursive values.
    For this last task, they modify existing techniques for the semi-dynamic setting where an array is indexed for amortized- and constant-time range queries in an online fashion.

    \item[Algorithms on EDSes.] 
    There is a wide literature on exact \citep{GrossiILPPRRVV17,AoyamaNIIBT18,DBLP:conf/wea/PissisR18,sopang,DBLP:conf/icalp/0001GPPR19,IliopoulosKP21,Beretal22,Ascetal24} and approximate \citep{Beretal17,Beretal20,latin22,Beretal24} pattern search on EDSes, on their alignment with linear strings \citep{MP22,DBLP:conf/biostec/MwanikiGP23,ieee2025}, and on the pairwise comparison of (elastic) degenerate strings \citep{DBLP:conf/wabi/AlzamelA0GIPPR18,Alzetal20,DBLP:conf/cpm/GaboryMPPRSZ23,front24,GMPPRSZ25} where algorithmic results exhibit complexity that depend on the size of EDSes.
    Moreover, the latest algorithms for 
    EDS-based pangenome comparison (intersection, matching statistics, similarity and distance measures) also depend on the cardinality of EDSes~\citep{DBLP:conf/cpm/GaboryMPPRSZ23,front24,GMPPRSZ25}. 
    In practical EDS applications, experimental heuristics are used to build EDSs from called variants (VCF files) \citep{sopang,DBLP:journals/bioinformatics/DanecekAAABDHLMSMD11,BOO23,GMPPRSZ25} and from MSAs \citep{GrossiILPPRRVV17,front24}, giving further motivation for our focus on optimizing these measures.

    \item[Indexable Founder Graphs.] \emph{Elastic Founder Graphs} (EFGs), introduced by  \cite{DBLP:conf/isaac/EquiNACTM21,Equi23foundergraphs}, take the simplified recombination model of founder sequences (recombination at selected positions only) to transform an MSA with insertions and deletions into an acyclic Elastic Block Graph: the segmentation results into consecutive blocks of nodes.
    In the general case, the resulting graph is still hard to index just like more general graphs \citep{DBLP:conf/icalp/EquiGMT19,DBLP:conf/sofsem/EquiMT21,EMTG23,EMT23}, so Equi et al.\ proved that if each chosen segment spells strings that exclusively occur from the segments' starting column, then this property is conserved in the resulting graph, which is easy to index. This results in a framework of linearithmic-time (i.e.\ $O(n \log n)$) construction algorithms for indexable EFGs, that support fast pattern matching \citep{DBLP:journals/bioinformatics/NorriCDVM21}. 
    
    \cite{Rizzo24efg} studied the problem of 
    minimizing the maximum segment length or maximum height of a segmentation resulting in an indexable EFG both in the gapless setting and in the setting with gaps~\citep{DBLP:conf/wabi/MakinenCENT20,DBLP:conf/iwoca/RizzoM22,Equi23foundergraphs,Rizzo24efg}, also improving the linearithmic construction time to linear.
    In developing an EFG-based sequence-to-graph pipeline for aligning reads based on the seed-chain-extend strategy, \cite{DBLP:journals/bioinformatics/0001CM25} relaxed EFGs into EDSes to efficiently chain fragments minimizing a metric based on edit distance, to find good approximate alignments to extend.

    \item[MSA covers.] While segmentation induces a limited class of acyclic pangenome graph representations, a more general approach to cover MSAs and generate arbitrary acyclic graphs was proposed by \cite{CartesBCVD24}; the related optimization problems are much harder, although small instances can be solved using integer linear programming.
\end{description}

\section{Preliminaries\label{sec:defs}}

We denote the integer interval $\lbrace x, x+1, \dots, y \rbrace$ as $[x..y]$, with $x,y \in \mathbb{Z}$.
Let $\msaij{1..r}{1..c} \in (\Sigma \cup \lbrace \gap \rbrace)^{r \times c}$ be a \emph{multiple sequence alignment} of $r$ rows and $c$ columns representing the input sequences (or \emph{strings}) and the aligned positions, respectively.
It is built from a finite integer alphabet $\Sigma = [1..\sigma]$ ($\Sigma = \lbrace \mathtt{A}, \mathtt{C}, \mathtt{G}, \mathtt{T} \rbrace$ in our examples, but it can be easily mapped to integers $[1..4]$) augmented with the \emph{gap symbol} $\gap \notin \Sigma$ to represent insertions and deletions.
We indicate the $i$-th row as $\msaij{i}{1..c} \in (\Sigma \cup \lbrace \gap \rbrace)^{c}$, and its $j$-th symbol as $\msaij{i}{j} \in \lbrace \Sigma \cup \gap \rbrace$.
Assuming to work in the standard RAM model of computation, $\sigma$ fits into a word and the alignment takes at most $O(rc)$ words of space, or $O(rc \log |\Sigma|)$ bits of space.
We say that $\msaij{1..r}{1..c}$ is \emph{gapless} if no gap symbol is used.
The operator $\spell(T)$ removes all the gap symbols from a given string $T \in (\Sigma \cup \lbrace \gap \rbrace)^c$ (e.g.\ $\spell(\mathtt{AGC}\gap\mathtt{{TT}-}) = \mathtt{AGCTT}$); if $T$ contains only gaps then $\spell(T) = \varepsilon$, the empty string.
Consider the following definition alongside \Cref{fig:example}.

\begin{definition}[Segmentation metrics]
    Given $\msaij{1..r}{1..c} \in (\Sigma \cup \lbrace \gap \rbrace)^{r \times c}$, a \emph{segmentation} of the MSA is a partition $S = S_1, S_2, \dots, S_k$ of $[1..c]$, that is, a sequence of $k$ contiguous and non-overlapping intervals covering $[1..c]$.
    In symbols,
\begin{gather*}
    S_1 = [x_1..y_1], \; \dots, \; S_k = [x_k..y_k]
    \\\enspace\text{with}\enspace
    x_1 = 1, \;
    y_{i} + 1 = x_{i+1} \; \forall i \in [1..k-1], \;\text{and}\;
    y_k = c.
\end{gather*}
    We denote (see Figure~\ref{fig:example} for examples):
\begin{itemize}[nosep] 
    \item the \emph{number of segments} of $S$ as $\lvert S \rvert = k$;
    \item the \emph{height} $\max_{i \in [1..k]} \segmentheight(S_i)$ of $S$ as $\maxheight(S)$, where the height $\segmentheight(S_i)$ of a segment $S_i = [x..y]$ is defined as
    \[
        \segmentheight([x..y]) =
        \Big\lvert \Big\lbrace
            \spell(\msaij{i}{x..y})
            \;\Big|\;
            i \in [1..r]
        \Big\rbrace \Big\rvert
    \]
    (note that the empty string is counted by $\segmentheight$);
    \item the \emph{cardinality} $\sum_{i \in [1..k]} \segmentheight(S_i)$ as $\stringsize(S)$;
    \item the \emph{size} $\sum_{i \in [1..k]} \size(S_i)$ as $\size(S)$, where the size $\size(S_i)$ of a segment $S_i = [x..y]$ is defined as
    \begin{gather*}
        \size([x..y]) =
        \sum_{T \in \spell(\msaij{1..r}{x..y})} \lvert T \rvert
        \qquad\text{with}\qquad \\
        \spell(\msaij{1..r}{x..y}) = \Big\lbrace \spell(\msaij{i}{x..y}) \quad\Big|\quad i \in [1..r] \Big\rbrace
    \end{gather*}
    (note that the empty string $\varepsilon$ contributes $0$ units to $\size([x..y])$); and
    \item the \emph{gap-aware size} $\gapsize(S)$ as above, substituting $\size([x..y])$ with
\[
    \gapsize([x..y]) = \begin{cases}
        \size([x..y]) + 1 &\text{if $\varepsilon \in \spell(\msaij{1..r}{x..y})$,} \\
        \size([x..y])     &\text{otherwise}.
    \end{cases}
\]
\end{itemize}
\end{definition}

Consider the problem of finding a segmentation of minimum cardinality.
If the MSA contains few but very long sequences, that is, $r \ll c$, the trivial segmentation $\trivialsegmentationhorizontal = [1..c]$ might be optimal.
We can encourage a certain level of recombination in the corresponding EDS by setting an upper bound $U$ on the length of the accepted segments.
\begin{problem}[min-$U$-cardinality]
    \label{prob:minUcard}
    Given $\msaij{1..r}{1..c}$ and an upper bound $U$ on the segment length, find a segmentation $S$ containing segments of length at most $U$ and minimizing the cardinality $\stringsize(S)$.
\end{problem}

\noindent On the other hand, consider minimizing the size.
In \Cref{subsec:trivial} we show that the trivial segmentation $\trivialsegmentationvertical = \lbrace 1 \rbrace, \dots, \lbrace c \rbrace$ has always minimum size if the MSA is gapless, resulting in a highly recombinant EDS.
We can symmetrically discourage recombination in regions of low sequence similarity by fixing a lower bound $L$ on the length of the accepted segments.
\begin{problem}[min-$L$-size]
    \label{prob:minLsize}
    Given $\msaij{1..r}{1..c}$ and a lower bound $L$ on the segment length, find a segmentation $S$ containing segments of length at least $L$ and minimizing the gap-aware size $\gapsize(S)$.
\end{problem}

The linear-time solutions we develop in \Cref{sec:algs} rely on efficient data structures for indexing aligned sequences and range minimum queries.
The first problem is solved by the positional Burrows--Wheeler transform, which
is based on iteratively sorting the MSA rows while reading them column by column.
Let $\Sigma$ have an implicit total order $\le$ on its characters.
Given strings $S,T \in \Sigma^c$, we say that $S$ is \emph{lexicographically smaller} than $T$, in symbols $S \preclex T$, if $S[1..x] = T[1..x]$ and $S[x+1] < T[x+1]$ for some $x \in [0..c-1]$ (where $S[1..0]$ is equal to the empty string $\varepsilon$).
Symmetrically, $S$ is \emph{co-lexicographically smaller} than $T$, in symbols $S \preccolex T$, if $S[x..c] = T[x..c]$ and $S[x-1] < T[x-1]$ for some $x \in [1..c+1]$ (where $S[c+1..c] = \varepsilon$).

\begin{definition}[Positional Burrows--Wheeler transform~\citep{Dur14}]
    \label{def:pBWT}
    The \emph{positional Burrows--Wheeler transform} (pBWT) of a \emph{gapless} $\msaij{1..r}{1..c}$ for position $x \in [1..c]$ consists of array $\mathtt{a}_x[1..r]$, representing the prefixes $\msaij{i}{1..x}$ for $i \in [1..r]$ sorted co-lexicographically, and array $\mathtt{d}_x[1..r]$, describing the length of the longest common suffixes of adjacent prefixes in the sorted order.
    More specifically:
\begin{itemize}[nosep]
    \item $\mathtt{a}_x[1..r]$ is the permutation\footnote{This permutation is unique if we assume each input sequence $\msaij{i}{1..c}$ to start with a unique character $\mathtt{s}_i$ marking its start, with $i \in [1..r]$.} of $[1..r]$ such that
    \[
        \msaij{\mathtt{a}_{x}[1]}{1..x}
        \preceqcolex
        \msaij{\mathtt{a}_{x}[2]}{1..x}
        \preceqcolex
        \dots
        \preceqcolex
        \msaij{\mathtt{a}_{x}[r]}{1..x};
    \]
    \item $\mathtt{d}_x[i]$ is an integer in $[1..x+1]$ such that the longest common suffix between $\msaij{\mathtt{a}_{x}[i]}{1..x}$ and $\msaij{\mathtt{a}_{x}[i-1]}{1..x}$ has length $\mathtt{d}_x[i]$, if $i > 1$ and such suffix is non-empty, otherwise $\mathtt{d}_x[i] = x + 1$.
\end{itemize}
\end{definition}

As shown by \cite{Dur14}, the pBWT allows an efficient incremental computation of arrays $\mathtt{a}$ and $\mathtt{d}$:

\begin{lemma}[\citep{Dur14,MakinenN19}]
    Let $\msaij{1..r}{1..c}$ be a gapless MSA over alphabet $\Sigma$ of size $O(r)$.
    Given the positional Burrows--Wheeler transform for position $x \in [1..{c-1}]$, that is, arrays $\mathtt{a}_x[1..r]$ and $\mathtt{d}_x[1..r]$, the arrays $\mathtt{a}_{x+1}[1..r]$ and $\mathtt{d}_{x+1}[1..r]$ can be computed in $O(r)$ time.
\end{lemma}

\noindent Finally, in \Cref{subsec:cardinality,subsec:size} we index an integer array to find out the minimum value in a given range of positions of the array with the following result from Cazaux et al.
\begin{lemma}[{Semi-dynamic range minimum query data structure for arrays~\citep[Lemma 7]{Cazaux19columnstream}}]
    \label{lem:RMQ}
    There exists a data structure that maintains an integer array $\mathtt{I}[1..c]$ and supports: the \emph{append} query, which adds a new element to the end of $\mathtt{I}$ and increments $c$ by one, in $O(1)$ amortized time; and the \emph{Range Minimum Query} which, for any given $x,y \in [1..c]$ with $x \le y$, computes a position $k \in [1..c]$ such that $\mathtt{I}[k] = \min \lbrace Q[\ell] : x \le \ell \le y \rbrace$, in $O(1)$ time.
\end{lemma}

\begin{lemma}[{Semi-dynamic range minimum query data structure for queues~\citep[Lemma 8]{Cazaux19columnstream}}]
    \label{lem:RMQueue}
    There exists a data structure that maintains an integer array $\mathtt{Q}[1..c]$ and supports: the \emph{append} query in $O(1)$ amortized time; the \emph{dequeue} query, which removes the element at the beginning of $\mathtt{Q}$ and decreases $c$ by one, in $O(1)$ amortized time; and the \emph{Range Minimum Query} in $O(1)$ time.
\end{lemma}

\section{Solutions for gapless MSAs}\label{sec:algs}
\label{sec:solutions}
In this section, we first show basic properties of gapless MSAs and that in the setting without gaps, \textsf{min-1-size} (i.e.\ no segment lower bound) is trivial whereas \textsf{min-$c$-cardinality} (i.e.\ no segment upper bound) is not (\Cref{subsec:trivial}).
Then, we develop $O(rc)$-time algorithms for \textsf{min-$U$-cardinality} (\Cref{subsec:cardinality}) and \textsf{min-$L$-size} (\Cref{subsec:size}) for gapless MSAs.

\subsection{Basic properties and non-trivial settings}
\label{subsec:trivial}
Before investigating under which setting problems \textsf{min-$U$-cardinality} and \textsf{min-$L$-size} are trivial, consider the following properties of MSAs without gaps that we will exploit in this section.
\begin{observation}[Monotonicity of left extensions~\citep{Norri19scalablefounder}]\label{obs:leftmonotonicity}
Given $\msaij{1..r}{1..c}$, for any $1 \le x \le y \le c$ we say that $[x..y]$ is a \emph{left extension} of suffix $\msaij{1..r}{y+1..c}$.
If the MSA is \emph{gapless}, the following monotonicity property holds for any left extension $[x..y]$:
\[
    r \ge \segmentheight([x'\!..y]) \ge \segmentheight([x..y]) \qquad \forall x' < x.
\]
The \emph{monotonicity of right extensions}, defined symmetrically as $\segmentheight([x..y]) \le \segmentheight([x..y']) \le r$ for any segment $[x..y]$ and $y' \in [y+1..c]$, also holds in the gapless case.
\end{observation}

\begin{observation}
    \label{obs:gaplesssize}
    Given a \emph{gapless} $\msaij{1..r}{1..c}$, the (gap-aware) size of any segment $[x..y]$ is equal to the height times the length of the segment, in symbols $\gapsize([x..y]) = \size([x..y]) = \lvert x..y \rvert \cdot \segmentheight([x..y])$.
    Indeed, in this setting $\spell(\msaij{i}{x..y}) = \msaij{i}{x..y}$ for all $i \in [1..r]$ and all strings spelled by segment $[x..y]$ have length $y - x + 1$.
\end{observation}

The simplest version of \textsf{min-$L$-size} occurs when the MSA contains no gap symbol and there is no lower bound on the segment length, that is, when $L = 1$.
Intuitively, this setting presents the same trivial solution as in the founder reconstruction problem [see \citep[Theorem 1]{Ukkonen02founder}], since transforming each column into its own segment achieves the best possible compression.

\begin{theorem}\label{theo:trivialsize}
    Let $[x..y]$, $[y+1..z]$ be adjacent segments of a \emph{gapless} $\msaij{1..r}{1..c}$. Then the size $\size([x..z])$ of their union is bigger or equal to their cumulative size $\size([x..y]) + \size([y+1..z])$,
    and an optimal solution to \textsf{min-$1$-size} for a \emph{gapless} $\msaij{1..r}{1..c}$ is the trivial segmentation $S^{\,\rotatebox[origin=c]{90}{$\equiv$}} = \lbrace 1 \rbrace, \lbrace 2 \rbrace, \dots, \lbrace c \rbrace$.
\end{theorem}
\begin{proof}
    Consider integers $1 \le x \le y < z \le c$ and their corresponding segments $[x..y]$, $[y+1..z]$, $[x..z]$.
    Then splitting $[x..z]$ into $[x..y]$ and $[y+1..z]$ never increases the size, since the MSA is gapless:
\begin{align*}
    \size([x..y]) + \size([y+1..z]) &= \lvert [x..y] \rvert \cdot \segmentheight([x..y]) + \lvert [y+1..z] \rvert \cdot \segmentheight([y+1..z]) &\text{(\Cref{obs:gaplesssize})}\\
    &\le \lvert [x..y] \rvert \cdot \segmentheight([x..z]) + \lvert [y+1..z] \rvert \cdot \segmentheight([x..z]) &\text{(\Cref{obs:leftmonotonicity})}\\
    &= \lvert [x..z] \rvert \cdot  \segmentheight([x..z])\\ &= \size([x..z]) &\text{(\Cref{obs:gaplesssize})}
\end{align*}
    Thus, the size of $S^{\,\rotatebox[origin=c]{90}{$\equiv$}}$ is less than or equal to the size of any other segmentation.\hfill\qedsymbol
\end{proof}

However, the trivial solution is not valid if $L > 1$ or it can be suboptimal if the MSA contains gaps: in the former setting, picking mostly segments of length exactly $L$ can be suboptimal, and in the latter case each empty string $\varepsilon$ contributes $1$ to the total size of the segmentation so the parts of the MSA with long runs of gaps need to be carefully considered.
See \Cref{fig:nontrivialminsize}.

\begin{figure}[htp]
    \centering
    \includegraphics{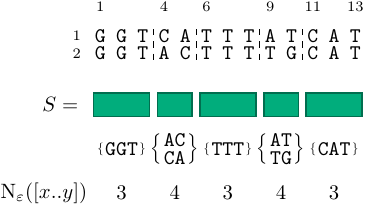}%
    \hfill%
    \includegraphics{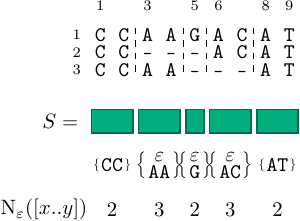}
    \caption{On the left, a gapless MSA where the optimal solution to \textsf{min-$2$-size} (i.e.\ the lower bound on the minimum segment length is $2$) has size 17 and is non-trivial. On the right, an MSA with gaps for which the optimal solution to \textsf{min-$1$-size} is non-trivial and has gap-aware size 12.}
    \label{fig:nontrivialminsize}
\end{figure}

Regarding the problem of minimizing the cardinality of the resulting segmentation, \textsf{min-$U$-card} behaves differently than other segmentation problems, as it does not admit a trivial optimal segmentation even if there is no upper bound on the segment length (i.e.\ $U = c$) and the MSA does not have any gap symbol.
\begin{observation}
    Trivial segmentations $S^{\,\rotatebox[origin=c]{90}{$\equiv$}} = \lbrace 1 \rbrace, \lbrace 2 \rbrace, \dots, \lbrace c \rbrace$ and $S^{\equiv} = [1..c]$ are not optimal with respect to \textsf{min-$c$-card} for some \emph{gapless} $\msaij{1..r}{1..c}$: on one hand, in regions of high similarity longer segments can be preferable as they can spell very few strings; on the other hand, shorter segments can generate less strings in regions of high diversity.
    See the example $\msaij{1..8}{1..4}$ in \Cref{fig:nontrivialmincard}.
\end{observation}
\begin{figure}[htp]
    \centering
    \includegraphics{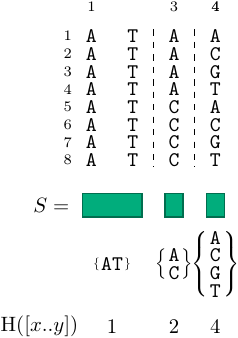}
    \caption{Example of a gapless $\msaij{1..8}{1..4}$ for which the trivial segmentation $S^{\equiv} = [1..c] = [1..4]$ is not a solution to \textsf{min-$4$-size}, since its cardinality $\stringsize(S^{\equiv})$ is equal to $8$. Instead, $S = [1..2], \lbrace 3 \rbrace, \lbrace 4 \rbrace$ (shown) is such that $\stringsize(S) = 7$.}
    \label{fig:nontrivialmincard}
\end{figure}

\subsection{Minimizing the cardinality}
\label{subsec:cardinality}
Given $\msaij{1..r}{1..c}$ and $U \in [1..c]$, for any $y \in [1..c]$ we define $\stringsize_y$ as the size of an optimal solution of \textsf{min-$U$-cardinality} on instance $\msaij{1..r}{1..y}$ respecting segment upper bound $U$.
Then, the following recursion holds, since the cardinality of a segmentation $S$ is the sum of the individual cardinalities of its segments:
\begin{align}
    \stringsize_0 &= 0, \nonumber \\
    \stringsize_y &= \min_{x \in [\max(1,\,y-U+1)..y]} \Big( \stringsize_{x-1} + \segmentheight([x..y]) \Big) & y \in [1..c], \label{eq:recstringsize}
\end{align}
and it is easy to see that $\stringsize_c$ is equal to the cardinality $\stringsize(S)$ of an optimal segmentation $S$.

\cite{Norri19scalablefounder} initially observed that, at least for gapless MSAs, all changes in segment height can be described compactly and in a range fashion by fixing an end column $y$ and considering values $\segmentheight([x..y])$ for $x \in [1..y]$.
This concept was restated and extended to MSAs with gaps by \cite{Norri19scalablefounder} and \cite{Rizzo24efg} into the \emph{extensions} of suffix $\msaij{1..r}{y+1..c}$ that are \emph{meaningful}.
Consider the following definition alongside \Cref{fig:meaningfulleftextensions}.
\begin{definition}[Meaningful left extensions~\citep{Rizzo24efg}]\label{def:lextextensions}
Given $\msaij{1..r}{1..c}$, let $L,U \in [1..c]$ be some given lower bound and upper bound on the segment length,\footnote{For simplicity, this definition accepts both a lower bound and an upper bound on the segment length, even though \minUcard\ and \minLsize\ require only one of the two. If either is unspecified, we simply assume that $L = 1$ or $U = c$ accordingly, that is, there is no lower or upper bound.} with $L \le U$.
For any $y \in [1..c]$ we denote with $\mathcal{L}_{y} = \ell_{y,1}, \dots, \ell_{y,d_y}$ the \emph{meaningful left extensions} of $\msaij{1..r}{y+1..c}$, meaning the strictly decreasing sequence of all positions smaller than or equal to $y$ such that:
\begin{enumerate}
    \item $y - U + 1 \le \ell_{y,d_y} < \dots < \ell_{y,2} < \ell_{y,1} = y - L + 1$, so that $\mathcal{L}_y$ captures the left extensions of $\msaij{1..r}{y+1..c}$ of length at least $L$ and at most $U$;
    \item $\segmentheight([\ell_{y,j}..y]) \neq \segmentheight([\ell_{y,j} + 1..y])$ for $2 \le j \le d_y$, so that each $\ell_{y,j}$ marks a column where the height of the left extension changes.
\end{enumerate}
If $y < L$, that is, $\msaij{1..r}{y+1..c}$ has no left extension of size at least $L$, we define $\mathcal{L}_y = ()$ and $d_y = 0$.
Otherwise, for completeness, we define $\ell_{y,d_y+1} = \max(0, y - U)$ (see how this is used in \Cref{eq:recmeaningfulcardinality}).
\end{definition}

\begin{figure}[htp]
    \centering
    \includegraphics{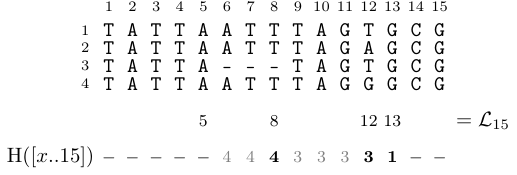}
    \caption{Example of $\msaij{1..4}{1..15}$ and of meaningful left extentions $\mathcal{L}_{15} = 13, 12, 8$ (represented from right to left in the figure), given lower bound $L = 3$ and upper bound $U = 10$.}
    \label{fig:meaningfulleftextensions}
\end{figure}

The meaningful left extensions $\mathcal{L}_y$ are indeed a compact description of all segment heights: all $[x..y]$ with $x \in [\ell_{y,j+1}+1..\ell_{y,j}]$ have height equal to $\segmentheight([\ell_{y,j}..y])$;
if the monotonicity of left extensions (\Cref{obs:leftmonotonicity}) holds, then $d_y \le r$, that is, $\lvert \mathcal{L}_y \rvert$ is at most $r$ for each $y \in [1..c]$ and the total number of meaningful left extensions is $O(rc)$.
Regardless of the number of meaningful left extensions, \Cref{eq:recstringsize} can be rewritten as
\begin{align}
    \stringsize_y 
    &= \min_{j \in [1..d_y]} \Big( \segmentheight([\ell_{y,j}..y]) + \min_{x \in [\ell_{y,j+1}+1..\ell_{y,j}]} \stringsize_{x-1} \Big). \label{eq:recmeaningfulcardinality}
\end{align}
Given pairs $(\ell_{y,j},\segmentheight([\ell_{y,j}..y]))$ in input for each $y \in [1..c]$, $j \in [1..d_y]$, we can compute values $\stringsize_y$ in a dynamic programming fashion and we can solve operation $\min_{x \in [\ell_{y,j+1}+1..\ell_{y,j}]} \stringsize_{x-1} = \min_{x \in [\ell_{y,j+1}..\ell_{y,j}-1]} \stringsize_{x}$
by indexing the array containing values $\stringsize_x$ for Range Minimum Queries with \Cref{lem:RMQ}, as shown in \Cref{alg:one}.

\begin{algorithm}[htp]
\caption{Segmentation of $\msaij{1..r}{1..c}$ with segment-length upper bound $U$ minimizing the cardinality. Note that bound $U$ is implicitly used in value $\ell_{y,d_y+1} = \max(0, y - U)$ (see \Cref{def:lextextensions} and \Cref{eq:recmeaningfulcardinality}).}\label{alg:one}
\Input{integers $r,c \in \mathbb{N}$,\\ segment length upper bound $U \in [1..c]$,\\ meaningful left extensions $(\ell_{y,j},h_{y,j})$ for $y \in [1..c]$, $j \in [1..d_y]$}
\Output{minimum cardinality of a segmentation $S_1, \dots, S_k$ such that $\lvert S_i \rvert \le U$ for $i \in [1..k]$}
Create array $\mathtt{m}[0..c]$ holding values in $[0..rc]$\;
Preprocess $\mathtt{m}$ for semi-dynamic Range Minimum Queries (\Cref{lem:RMQ})\;
$\mathtt{m}[0] \gets 0$\;
\For{$y \gets 1$ \KwTo $c$}{%
    $\mathtt{m}[y] \gets +\infty$\;
    \For{$j \gets 1$ \KwTo $d_y$}{%
        $x \gets \mathrm{RangeMinQuery}\big(\mathtt{m}, [\ell_{y,j+1}..\ell_{y,j}-1]\big)$\;
        $\mathtt{m}[y] \gets \min\Big( \mathtt{m}[y], \quad h_{y,j} + \mathtt{m}[x] \Big)$\tcp*{\Cref{eq:recmeaningfulcardinality}}
    }
    Update $\mathtt{m}$ for Range Minimum Queries over incremented range $[1..y]$\;
}
\Return{$\mathtt{m}[c]$}
\end{algorithm}

For gapless MSAs, the meaningful left extensions can be computed efficiently with a modification of the positional Burrows--Wheeler Transform (\Cref{def:pBWT}).
Indeed, \cite{Norri19scalablefounder} correctly realized that the pBWT for position $y$ already describes $\mathcal{L}_{y}$ and all changes in height (see \citep[Lemma 4]{Norri19scalablefounder}) and modified the pBWT to obtain values $(\ell_{y,j},\segmentheight([\ell_{y,j}..y]))$ in $O(r)$ time per each column $y$.
\begin{lemma}[{\citep[Lemmas 5 and 6]{Norri19scalablefounder}}]
    \label{lem:meaningfulleftcomputation}
    Let $\msaij{1..r}{1..c}$ be a gapless multiple sequence alignment over alphabet $\Sigma$, with $\lvert \Sigma \rvert \in O(r)$.
    The meaningful left extensions $\mathcal{L}_y = \ell_{y,1}, \dots, \ell_{y,d_y}$ and their corresponding segment height values $\segmentheight([\ell_{y,j}..y])$ for $j \in [1..d_y]$ can be computed for $y = 1, \dots, c$ in $O(rc)$ time and in $O(r + c)$ working space.
    Specifically, the values are computed column-wise in a streaming fashion, from $y = 1$ to $y = c$, and the space bound does not include storing the computed values.
\end{lemma}

Finally, in the gapless case we can simplify \Cref{eq:recmeaningfulcardinality} to $
\stringsize_y 
    = \min_{j \in [1..d_y]} ( \segmentheight([\ell_{y,j}..y]) +\stringsize_{\ell_{y,j+1}} )
$ due to the following fact.
\begin{lemma}
    \label{lem:mmonotone}
    Let $\msaij{1..r}{1..c}$ be a gapless multiple sequence alignment. 
    The values $\stringsize_y$, the sizes of optimal solutions to \minUcard\ ending at columns $y \in [1..c]$ (\Cref{eq:recstringsize}), are non-decreasing.
\end{lemma}
\begin{proof}
    Consider column $y \in [2..c]$ and an optimal segmentation $S_y = [x_1..y_1],  \dots, [x_k..y]$ for \textsf{min-$U$-cardinality} on instance $\msaij{1..r}{1..y}$ respecting segment upper bound $U$, that is, $m_y$ = $m(S_y)$. Next, we construct segmentation $S_{y-1}$ by removing the last column from $S_y$, and we prove that is has smaller or equal cardinality as follows. If the last segment has length 1, simply removing it results in a lower cardinality segmentation. Otherwise $S_{y-1} = [x_1..y_1],  \dots, [x_k..y-1]$, with the last segment having a lower height $\segmentheight([x_k..y-1]) \le \segmentheight([x_k..y])$ due to the monotonicity of right extensions (\Cref{obs:leftmonotonicity}). Therefore, $m(S_{y-1}) \le m(S_y)$, while still respecting the segment upper bound $U$. $S_{y-1}$ may or may not be optimal, so $m_{y-1} \le m(S_{y-1}) \le m(S_y) = m_y$. \hfill\qedsymbol
\end{proof}

The computation of \Cref{eq:recmeaningfulcardinality} is carried by \Cref{alg:one} after the computation of the meaningful left extensions following \cite[Section 4.6]{Rizzo24efg}, or the computation from \Cref{lem:meaningfulleftcomputation} for the gapless case, obtaining the following result.
\begin{theorem}
    \label{theo:mincard}
    Let $\msaij{1..r}{1..c}$ be a multiple sequence alignment over alphabet $\Sigma \cup \lbrace \gap \rbrace$, and let $U \in [1..c]$ be an upper bound on the maximum segment length.
    We can solve \minUcard\ (\Cref{prob:minUcard}) in $O(r \cdot c \cdot U \cdot \log \lvert \Sigma \rvert)$ time and $O(c + r \cdot U)$ working space, by finding the optimal segmentation $S$ respecting $U$ and minimizing the cardinality $\stringsize(S)$.
    If the MSA is gapless and  $\lvert \Sigma \rvert \in O(r)$, then we can solve \minUcard\ in $O(rc)$ time and $O(c + r)$ working space.
\end{theorem}
\begin{proof}
    The final algorithm is provided by combining \Cref{alg:one}, that utilizes the meaningful left extensions $\mathcal{L}_y$ and their height values from $y = 1$ to $y = c$, with a computation of the meaningful left extensions providing pairs $(\ell_{y,j},\segmentheight([\ell_{y,j}..y]))$ in the same order.
    The correctness of \Cref{alg:one} follows from \Cref{eq:recmeaningfulcardinality}.
    In the general case with gaps, \cite[Lemma 6]{Rizzo24efg} provides a $O(r \cdot c \cdot U \cdot \log \lvert \Sigma \rvert)$-time computation of the extensions (right extensions, in the paper) using dynamic keyword trees, reaching the stated time and space complexity since the complexity of this computation is always an upper bound of the time spent by \Cref{alg:one}.
    When the MSA is gapless, the total number of iterations of the inner for-loop in \Cref{alg:one} is $O(rc)$ due to the monotonicity of left extensions (\Cref{obs:leftmonotonicity}), and we can even skip using the Range Minimum Query data structure due to \Cref{lem:mmonotone}.
    \Cref{alg:one} can be augmented with standard backtracking techniques to obtain the actual segmentation, as shown by \cite[Section 4.5]{Rizzo24efg}. \hfill\qedsymbol
\end{proof}

\subsection{Minimizing the size}
\label{subsec:size}
Similarly to \Cref{subsec:cardinality}, given $\msaij{1..r}{1..c}$ and $L \in [1..c]$, for any $y \in [1..c]$ we define $\size_y$ as the size of an optimal solution of \textsf{min-$L$-size} on instance $\msaij{1..r}{1..y}$ respecting lower bound $L$.
Then the following recursion holds, since the gap-aware size of a segmentation $S$ is the sum of the individual gap-aware segment sizes:
\begin{align}
    \size_0 &= 0, \nonumber \\
    \size_y &= \min_{x \in [1..y-L+1]} \Big( \size_{x-1} + \gapsize([x..y]) \Big) & y \in [1..c],\label{eq:recsize}
\end{align}
where $[1..y-L+1] = \emptyset$ and $\min(\emptyset) = +\infty$ when $y < L$, and it is easy to see that $\size_c$ is equal to the minimum gap-aware size $\gapsize(S)$ of a segmentation $S$ respecting $L$.

We now concentrate on gapless MSAs.
Recall that, in this setting, segment size and gap-aware segment size coincide, and they are equal to the length of the segment times its height (\Cref{obs:gaplesssize}).
Also, consider how the meaningful left extensions (\Cref{def:lextextensions}) are a compact description of all possible segment heights, as discussed in \Cref{subsec:cardinality}.
Then, in the gapless setting, \Cref{eq:recsize} for $y \in [1..c]$ can be rewritten as
\begin{align}
    \size_y &= \min_{x \in [1..y-L+1]} \Big( \size_{x-1} + \segmentheight([x..y]) \cdot (y - x + 1) \Big) & \text{(\Cref{obs:gaplesssize})} \nonumber\\
    &= \min_{x \in [1..y-L+1]} \Big( \segmentheight([x..y]) \cdot y + \size_{x-1} - \segmentheight([x..y]) \cdot (x-1) \Big)\nonumber\\
    &= \min_{j \in [1..d_y]} \Big( \segmentheight([\ell_{y,j}..y]) \cdot y \enspace+ 
    \min_{x \in [\ell_{y,j+1}..\ell_{y,j}-1]} \Big( \size_x - \segmentheight([\ell_{y,j}..y]) \cdot x \Big)  \Big), \label{eq:recmeaningfulsize}
\end{align}
where the last step holds because the meaningful left extensions $\mathcal{L}_{y}$ are defined to partition range $[1..y-L+1]$ according to the segment height $\segmentheight([x..y])$ for variable $x$, or in other words, $\segmentheight([x..y]) = \segmentheight([\ell_{y,j}..y])$ for all $x \in [\ell_{y,j+1}+1..\ell_{y,j}]$, $j \in [1..d_y]$.
\Cref{eq:recmeaningfulsize} is suitable for an efficient dynamic programming approach, shown in \Cref{alg:two}:
\begin{itemize}
    \item value $\segmentheight([\ell_{y,j}..y]) \cdot y$ inside the outer $\min$ operator depends on $y$ and on the height considered (i.e.\ the meaningful left extension), as in \Cref{eq:recmeaningfulcardinality} from \Cref{subsec:cardinality};
    \item the addends in the internal $\min$ operation depend on $x$ and on the segment height $\segmentheight([\ell_{y,j}..y]) \in [1..r]$, and since the number of possible height values is $O(r)$ we can store and maintain these recursive values for each column $x \in [1..y-1]$ and each height $h \in [1..r]$.
\end{itemize}

\begin{algorithm}[htp]
\caption{Segmentation of gapless $\msaij{1..r}{1..c}$ with segment-length lower bound $L$ minimizing the size. Note that bound $L$ is implicitly used in value $\ell_{y,1} = y-L+1$ (see \Cref{def:lextextensions}).}\label{alg:two}
\Input{integers $r,c \in \mathbb{N}$,\\ segment length lower bound $L \in [1..c]$,\\ meaningful left extensions $(\ell_{y,j},h_{y,j})$ for $y \in [1..c]$, $j \in [1..d_y]$}
\Output{minimum size of a segmentation $S_1, \dots, S_k$ such that $\lvert S_i \rvert \ge L$ for $i \in [1..k]$}
Initialize array $\mathtt{N}[1..c]$ with values in $[0..rc]$\;
Initialize matrix $\mathtt{M}[1..r,0..c]$ with values in $[0..rc]$\;
Preprocess the rows of $\mathtt{M}$ for semi-dynamic Range Minimum Queries (\Cref{lem:RMQ})\;
\For{$i \gets 1$ \KwTo $r$}{%
    $\mathtt{M}[i,0] \gets 0$\;
}
\For{$y \gets 1$ \KwTo $c$}{%
    $\mathtt{N}[y] \gets +\infty$\;
    \For{$j \gets 1$ \KwTo $d_y$}{%
        $x \gets \mathrm{RangeMinQuery}\big(\mathtt{M}[h_{y,j}], [\ell_{y,j+1}..\ell_{y,j}-1]\big)$\;
        $\mathtt{N}[y] \gets \min\Big( \mathtt{N}[y], \quad h_{y,j} \cdot y + \mathtt{M}[h_{y,j},x] \Big)$\tcp*{\Cref{eq:recmeaningfulsize}}
    }
    \For{$h \gets 1$ \KwTo $r$}{%
        $\mathtt{M}[h,y] \gets \mathtt{N}[y] - h \cdot y$\;
        Update Range Minimum Query data structure of row $h$ to cover incremented range $\mathtt{M}[h,1..y]$\;
    }
    
}
\Return{$\mathtt{N}[c]$}
\end{algorithm}

\begin{theorem}
    \label{theo:minsize}
    Given \emph{gapless} $\msaij{1..r}{1..c}$ and a lower bound $L \in [1..c]$ on the minimum segment length, we can compute in $O(rc)$ time and $O(c + L\cdot r)$ working space the optimal segmentation $S$ respecting $L$ and minimizing the size $\size(S)$.
\end{theorem}
\begin{proof}
    The final algorithm interleaves the computation of \Cref{lem:meaningfulleftcomputation}---pairs $(\ell_{y,j},h_{y,j})$ for $y \in [1..c]$ and $j \in [1..d_y]$---with \Cref{alg:two}, whose
    correctness follows from \Cref{eq:recmeaningfulsize}.
    The total number of iterations of the two inner for-loops is $O(rc)$, since the number of meaningful left extensions is $O(rc)$ (\Cref{obs:leftmonotonicity}) and the possible height values are in the range $[1..r]$.
    \Cref{alg:one} reaches the stated time complexity and uses $O(rc)$ space by applying \Cref{lem:meaningfulleftcomputation,lem:RMQ}.
    The working space can be improved from $O(r \cdot c)$ to $O(c + L \cdot r)$ by substituting matrix $\mathtt{M}$ with $r$ range minimum queues (\Cref{lem:RMQueue}) storing the recursive information for only the previous $2L$ columns (i.e.\ in range $[y-2L..y]$): due to \Cref{theo:trivialsize}, any segment with length greater or equal than $2L$ can be split into segments of length between $L$ and $2L - 1$ without increasing the total size, so an optimal solution using segments of length between $L$ and $2L$ always exists.
    Similarly to \Cref{alg:one}, \Cref{alg:two} can be augmented with standard backtracking techniques to obtain the actual segmentation. \hfill\qedsymbol
\end{proof}

\section{Experiments\label{sec:experiments}}

Before discussing our implementation and experiments, we make the following observation on the theoretical approach.
Even though \Cref{eq:recmeaningfulcardinality,eq:recsize} (but not \Cref{eq:recmeaningfulsize}) hold for MSAs with gaps, \Cref{theo:mincard,theo:minsize} work only on MSAs without gaps.
A straightforward way to extend the approach is to treat gaps as normal symbols in the alphabet, and remove them after the segmentation to obtain the final EDS.
For example, consider the height of a segment that contains \texttt{A-} and  \texttt{-A}: it is reduced by at least one after gaps are removed.

\begin{observation}\label{obs:gaps-as-symbols}
     Let $\msaij{1..r}{1..c} \in (\Sigma \cup \lbrace \gap \rbrace)^{r \times c}$ be a multiple sequence alignment and let $\mathrm{OPT} = \stringsize(S)$ be the cardinality of a solution to \minUcard.
     We can interpret $\msaij{1..r}{1..c}$ as \emph{gapless} by treating the gap symbol as a character (i.e. $\gap \in \Sigma$): let $\mathrm{OPT}_{\gap \in \Sigma}$ be the cardinality of an optimal solution in this setting.
     Then $\mathrm{OPT}_{\gap \in \Sigma} \ge \mathrm{OPT}$ so $\mathrm{OPT}_{\gap \in \Sigma}$ is an upper bound on the minimal cardinality of the original problem.
     The same argument holds for \minLsize.
 \end{observation}

Observing that the removal of gaps cannot increase the height, it is of interest to study how much the measures change after the gap removal in practice.
We developed a C++ implementation of \Cref{alg:one} finding $\mathrm{OPT}$ and $\mathrm{OPT}_{\gap \in \Sigma}$ as tool \texttt{mincard}, available at \url{https:github.com/algbio/eds}.
To optimize space usage, \texttt{mincard} indexes and streams the input MSA in chunks with HSTlib~\citep{bonfield2021htslib}. We implemented two approaches for the computation of the meaningful left extensions (\Cref{def:lextextensions}): (i)
$O(r \cdot c \cdot U)$ time preprocessing (assuming a constant-sized alphabet $\Sigma$) using keyword trees in order to enable computation of $\mathrm{OPT}$, where the number of meaningful left extensions is $\Theta(r \cdot U)$ in the worst case (see also \citep[Lemma 6]{Rizzo24efg}); and (ii) the $O(rc)$-time preprocessing of \Cref{theo:mincard} with the pBWT to enable the computation of $\mathrm{OPT}_{\gap \in \Sigma}$.
To show the versatility of the segmentation framework, \texttt{mincard} can additionally solve the variant of \Cref{eq:recstringsize} that considers not only the segments of length at most $U$, but also all \emph{perfect segments} representing fully conserved regions (i.e.\ columns where all input sequences express the same character), and we call this variant \texttt{mincard $U$ pc}.
In the following, we compare our methods to: the trivial segmentations $\trivialsegmentationvertical$ and $\trivialsegmentationhorizontal$ (from \Cref{sec:defs}); the heuristic segmentation $\trivialsegmentationhorizontal$ \texttt{pc} that selects the maximal perfect segments and falls back to $\trivialsegmentationhorizontal$ in-between; and the \texttt{msatoeds} tool\footnote{Available at \url{https://github.com/urbanslug/junctions/tree/master/scripts/msatoeds}} from \cite{front24}.
The latter is a greedy left-to-right method that picks segments of length at least 4 and stops extending a segment just before its height would increase.

We tested the single-threaded construction methods on an Intel(R) Xeon(R) Gold 6248 CPU @ 2.50 GHz with 376 GBs of RAM, in the University of Helsinki \texttt{kale} cluster.
The tested datasets are SARS-CoV-2 MSAs of different sizes, a simulated E.\ coli MSA, and an existing human MSA:
\begin{itemize}[nosep]
    \item four \emph{SARS-CoV-2} MSAs of sizes $[101, 29\,903]$, $[1001, 29\,903]$, $[10001, 29\,903]$, and $[100001, 29\,903]$, obtained by randomly sampling $10^2, \dots, 10^5$ complete sequences with no ambiguous characters from the National Center for Biotechnology Information (NCBI) Virus SARS-CoV-2 Data Hub\footnote{Available at \url{https://www.ncbi.nlm.nih.gov/sars-cov-2/}. The full list of 1\,106\,603 accession IDs was obtained on Apr 15th 2026 with filters `Nucleotide Completeness': complete and `Ambiguous Characters': 0.} and aligning them with \texttt{ViralMSA}~\citep{ViralMSA-bioinf} (and thus adding the \texttt{NC\_045512.2} reference);
    \item an \emph{E.\ coli} multiple alignment $\msaij{1..16}{1..6\,809\,339}$ simulated with \texttt{AliSim}~\citep{alisim}, starting from the strain \emph{K-12 MG1655} (assembly \texttt{GCF\_000005845.2}, downloaded from NCBI), based on a randomly generated phylogenetic tree with 16 leaves and mean branch distance 0.1, where on each branch 5\% of columns are modified according to the Jukes--Cantor substitution model and every 100 substitutions an indel of average length 1000 is introduced;\footnote{The command used is \texttt{iqtree {-}{-}alisim input/msa {-}{-}root-seq e\_coli\_K-12\_MG1655.fasta,NC\_000913.3 {-}{-}indel 0.01,0.01 \-\-indel-size "POW\{1.7/1000\},POW\{1.7/1000\}" -t "RANDOM\{yh/16\}" -m "JC+I\{0.95\}" -af fasta \-\-seed 302288}.}
    \item a human chromosome 19 multiple alignment $\msaij{1..1\,000}{1..59\,451\,290}$ aligned with tool \texttt{PANAMA}~\citep{olbrich2025generating} starting from the short variants\footnote{The variants in VCF format were used to obtain the original input sequences with \texttt{bcftools}~\citep{massimilianorossi}.} of the 1000 Genomes Project, originally from \cite{BoucherGIKLMNP021}.
\end{itemize}
The features of the MSAs are summarized in \Cref{tab:MSA-stats}.

\begin{table}[htp]
    \centering
    \caption{Statistics of the tested MSAs, where $r$ is the number of rows, $c$ is the number of columns, `bps' is the number of non-gap characters (in millions), `\% gaps' is the percentage of gaps in the total MSA size $r \cdot c$, `\% p.\ cols' is the percentage of perfect columns in the total column number $c$.}
    \vspace*{1ex}
    \begin{tabular}{lrrrrrrr}
        \toprule
        MSA & $r$ & $c$ & bps (M) & \% gaps & \% p.\ cols \\
        \midrule
        SARS-CoV-2-$10^2$ &    101 & 29\,903 &    3.0 & 0.49 & 95.97 \\
        SARS-CoV-2-$10^3$ &   1001 & 29\,903 &   29.8 & 0.47 & 85.16 \\
        SARS-CoV-2-$10^4$ &  10001 & 29\,903 &  297.6 & 0.49 & 50.49 \\
        SARS-CoV-2-$10^5$ & 100001 & 29\,903 & 2975.6 & 0.49 &  9.60 \\
        E.\ Coli (sim) &   16 &  6\,809\,339 &      74.3 &  31.84 & 41.65 \\
        Human chr19    & 1000 & 59\,451\,290 & 59\,065.1 &   0.65 & 97.71 \\
        \bottomrule
    \end{tabular}
    \label{tab:MSA-stats}
\end{table}

In \Cref{fig:covid-card-size,fig:ecoli_human}, we compare the EDS cardinalities and sizes (under the $\mathrm{OPT}_{\gap \in \Sigma}$ segmentation)  for \texttt{msatoeds} and \texttt{mincard} using different segment length upper bound values $U$.
Tool \texttt{msatoeds} exceeded the 376GB RAM limit on the human dataset.\footnote{Using a different cluster node with 1000GB of memory, \texttt{msatoeds} did not complete in under 24 hours on the human dataset.}
As conjectured in \Cref{sec:defs}, a smaller value of $U$ encourages compression and recombination, reducing the size.
Moreover, when the percentage of perfect columns is low, using the perfect columns is less effective, and when the number of input sequences is very high the non-trivial methods give a smaller cardinality than $\trivialsegmentationhorizontal$, with the exception of $\trivialsegmentationhorizontal\ \mathtt{pc}$.
The EDSes computed by treating gaps as symbols, finding $\mathrm{OPT}_{\gap \in \Sigma}$, segmenting, and only then removing the gaps, show a negligible degradation of the EDS features in practice: the cardinalities and sizes increase by at most $0.09\%$ and $0.25\%$, respectively (see \Cref{tab:opt} in \Cref{sect:running-times}, where we show the analogous results on $\mathrm{OPT}$ for small $U$).

\begin{figure}[htp]
    \centering
    \includegraphics{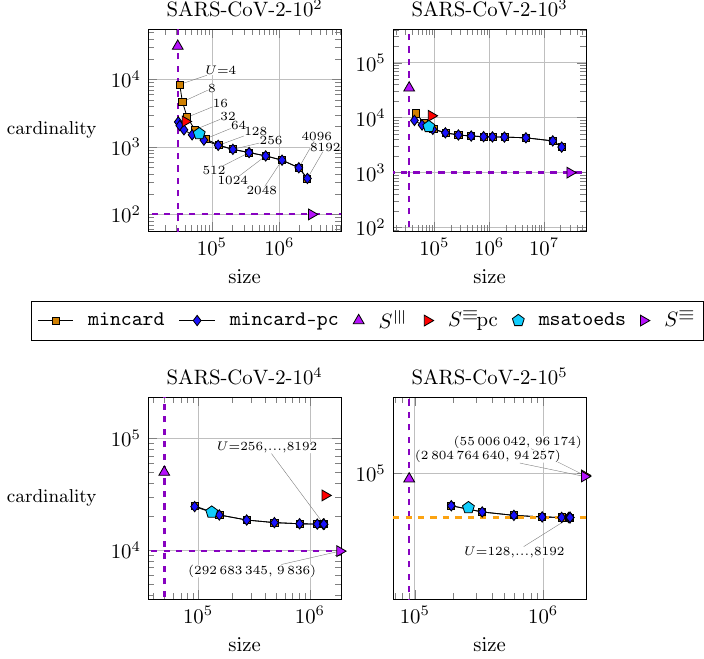}
    \caption{Log-log comparison of size and cardinality of different EDS construction algorithms on the SARS-CoV-2 datasets, where the trivial segmentations $\trivialsegmentationvertical$ (every column is a segment) and $\trivialsegmentationhorizontal$ (no segmentation) give heuristic lower bounds to the size and cardinality, respectively, marked by the dashed lines. Note that the axes do not start at $10^0$.}
    \label{fig:covid-card-size}
\end{figure}
\begin{figure}[htp]
    \centering
    \includegraphics{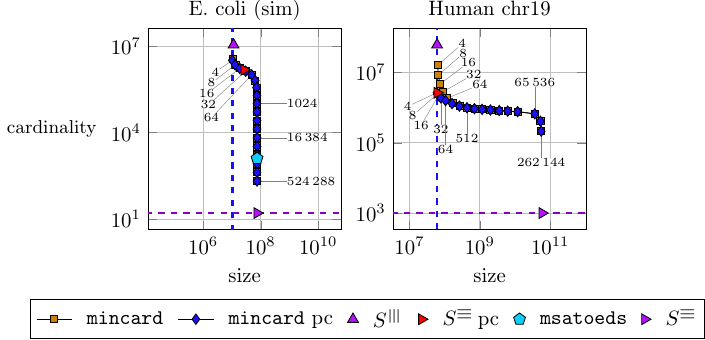}
    \caption{Log-log comparison of size and cardinality of different EDS construction algorithms, where the trivial segmentations $\trivialsegmentationvertical$ (every column is a segment) and $\trivialsegmentationhorizontal$ (no segmentation) give heuristic lower bounds to the size and cardinality, respectively, marked by the dashed lines. Note that the axes do not start at $10^0$.}
    \label{fig:ecoli_human}
\end{figure}

Figure~\ref{fig:runningtime-ecoli} illustrates the running times on part of the datasets. Here \texttt{mincard} is run in the mode that computes segmentations under $\mathrm{OPT}_{\gap \in \Sigma}$. On SARS-CoV-2 the maximum memory usage was 2.8 GB over runs on different parameters of \texttt{mincard}, 2.8 GB on $\trivialsegmentationvertical$, 2.8 GB on $\trivialsegmentationhorizontal$ \texttt{pc}, and 26 GB on \texttt{msatoeds}. On E.\ Coli the maximum memory usage was 225 MB over runs on different parameters of \texttt{mincard}, 111 MB on $\trivialsegmentationvertical$, 34 MB on $\trivialsegmentationhorizontal$ \texttt{pc}, and 1.7 GB on \texttt{msatoeds}. On Human chromose 19 the maximum memory usage was 1.9 GB over runs on different parameters of \texttt{mincard}, 1.0 GB on $\trivialsegmentationvertical$, and 163 MB on $\trivialsegmentationhorizontal$ \texttt{pc}, while \texttt{msatoeds} did not finish on this input.
While the running time of \texttt{mincard} run in this gaps as symbol mode is not asymptotically dependent on parameter $U$, some mild effect is visible due to the output segmentation depending on the parameter. The independence of \texttt{mincard} on parameter $U$ in this mode can be observed when comparing to the \texttt{mincard} in the mode that computes segmentations under $\mathrm{OPT}$, see \Cref{fig:scalability} in \Cref{sect:running-times}.

\begin{figure}
\centering
\resizebox{0.7\textwidth}{!}{%
\begin{tikzpicture}
\node[anchor=south,font=\normalsize\bfseries] at (9.65,0.5) {SARS-CoV-2-$\mathbf{10^5}$};
\node[anchor=east,font=\small] at (3.7500,0.0000) {\texttt{mincard}};
\draw (4.33959,0.00000) circle (0.13cm);
\fill (4.30417,0.00000) circle (0.08cm);
\fill (4.29291,0.00000) circle (0.08cm);
\fill (4.28922,0.00000) circle (0.08cm);
\fill (4.28821,0.00000) circle (0.08cm);
\fill (4.28542,0.00000) circle (0.08cm);
\fill (4.28627,0.00000) circle (0.08cm);
\fill (4.28571,0.00000) circle (0.08cm);
\fill (4.28620,0.00000) circle (0.08cm);
\fill (4.28617,0.00000) circle (0.08cm);
\fill (4.27909,0.00000) circle (0.08cm);
\node[inner sep=0pt,font=\large] at (4.28084,0.00000) {$\triangle$};
\node[anchor=east,font=\small] at (3.7500,-0.9500) {\texttt{mincard pc}};
\draw (4.35128,-0.95000) circle (0.13cm);
\fill (4.32353,-0.95000) circle (0.08cm);
\fill (4.31328,-0.95000) circle (0.08cm);
\fill (4.30368,-0.95000) circle (0.08cm);
\fill (4.30167,-0.95000) circle (0.08cm);
\fill (4.29688,-0.95000) circle (0.08cm);
\fill (4.30529,-0.95000) circle (0.08cm);
\fill (4.30352,-0.95000) circle (0.08cm);
\fill (4.29900,-0.95000) circle (0.08cm);
\fill (4.29946,-0.95000) circle (0.08cm);
\fill (4.28969,-0.95000) circle (0.08cm);
\node[inner sep=0pt,font=\large] at (4.29100,-0.95000) {$\triangle$};
\node[anchor=east,font=\small] at (3.7500,-1.9000) {$\trivialsegmentationhorizontal$};
\fill (4.20068,-1.90000) circle (0.08cm);
\node[anchor=east,font=\small] at (3.7500,-2.8500) {$\trivialsegmentationvertical$ \texttt{pc}};
\fill (4.05325,-2.85000) circle (0.08cm);
\node[anchor=east,font=\small] at (3.7500,-3.8000) {\texttt{msatoeds}};
\fill (15.00000,-3.80000) circle (0.08cm);
\draw[->] (4.0000,-4.5500) -- (15.3000,-4.5500);
\draw (4.00000,-4.61000) -- (4.00000,-4.49000);
\node[anchor=north,font=\scriptsize] at (4.00000,-4.65000) {0};
\draw (5.19051,-4.61000) -- (5.19051,-4.49000);
\node[anchor=north,font=\scriptsize] at (5.19051,-4.65000) {2000};
\draw (6.38103,-4.61000) -- (6.38103,-4.49000);
\node[anchor=north,font=\scriptsize] at (6.38103,-4.65000) {4000};
\draw (7.57154,-4.61000) -- (7.57154,-4.49000);
\node[anchor=north,font=\scriptsize] at (7.57154,-4.65000) {6000};
\draw (8.76205,-4.61000) -- (8.76205,-4.49000);
\node[anchor=north,font=\scriptsize] at (8.76205,-4.65000) {8000};
\draw (9.95257,-4.61000) -- (9.95257,-4.49000);
\node[anchor=north,font=\scriptsize] at (9.95257,-4.65000) {10000};
\draw (11.14308,-4.61000) -- (11.14308,-4.49000);
\node[anchor=north,font=\scriptsize] at (11.14308,-4.65000) {12000};
\draw (12.33359,-4.61000) -- (12.33359,-4.49000);
\node[anchor=north,font=\scriptsize] at (12.33359,-4.65000) {14000};
\draw (13.52411,-4.61000) -- (13.52411,-4.49000);
\node[anchor=north,font=\scriptsize] at (13.52411,-4.65000) {16000};
\draw (14.71462,-4.61000) -- (14.71462,-4.49000);
\node[anchor=north,font=\scriptsize] at (14.71462,-4.65000) {18000};
\node[anchor=north,font=\small] at (9.5000,-5.0000) {Running time (s)};

\end{tikzpicture}
}\\
\resizebox{0.7\textwidth}{!}{%
\begin{tikzpicture}
\node[anchor=south,font=\normalsize\bfseries] at (9.65,0.5) {E. Coli};
\node[anchor=east,font=\small] at (3.7500,0.0000) {\texttt{mincard}};
\draw (6.08118,0.00000) circle (0.13cm);
\fill (5.48328,0.00000) circle (0.08cm);
\fill (5.43345,0.00000) circle (0.08cm);
\fill (5.39512,0.00000) circle (0.08cm);
\fill (5.36063,0.00000) circle (0.08cm);
\fill (5.37213,0.00000) circle (0.08cm);
\fill (5.34146,0.00000) circle (0.08cm);
\fill (5.36446,0.00000) circle (0.08cm);
\fill (5.37213,0.00000) circle (0.08cm);
\fill (5.33763,0.00000) circle (0.08cm);
\fill (5.36446,0.00000) circle (0.08cm);
\fill (5.32230,0.00000) circle (0.08cm);
\fill (5.31080,0.00000) circle (0.08cm);
\fill (5.29930,0.00000) circle (0.08cm);
\fill (5.28014,0.00000) circle (0.08cm);
\fill (5.30697,0.00000) circle (0.08cm);
\fill (5.26864,0.00000) circle (0.08cm);
\node[inner sep=0pt,font=\large] at (5.38362,0.00000) {$\triangle$};
\node[anchor=east,font=\small] at (3.7500,-0.9500) {\texttt{mincard pc}};
\draw (5.92021,-0.95000) circle (0.13cm);
\fill (5.78223,-0.95000) circle (0.08cm);
\fill (5.74007,-0.95000) circle (0.08cm);
\fill (5.71324,-0.95000) circle (0.08cm);
\fill (5.68258,-0.95000) circle (0.08cm);
\fill (5.70557,-0.95000) circle (0.08cm);
\fill (5.68641,-0.95000) circle (0.08cm);
\fill (5.69791,-0.95000) circle (0.08cm);
\fill (5.67491,-0.95000) circle (0.08cm);
\fill (5.65192,-0.95000) circle (0.08cm);
\fill (5.65575,-0.95000) circle (0.08cm);
\fill (5.68641,-0.95000) circle (0.08cm);
\fill (5.68258,-0.95000) circle (0.08cm);
\fill (5.61359,-0.95000) circle (0.08cm);
\fill (5.61359,-0.95000) circle (0.08cm);
\fill (5.61359,-0.95000) circle (0.08cm);
\fill (5.59443,-0.95000) circle (0.08cm);
\node[inner sep=0pt,font=\large] at (5.65958,-0.95000) {$\triangle$};
\node[anchor=east,font=\small] at (3.7500,-1.9000) {$\trivialsegmentationhorizontal$};
\fill (5.87422,-1.90000) circle (0.08cm);
\node[anchor=east,font=\small] at (3.7500,-2.8500) {$\trivialsegmentationvertical$ \texttt{pc}};
\fill (4.78188,-2.85000) circle (0.08cm);
\node[anchor=east,font=\small] at (3.7500,-3.8000) {\texttt{msatoeds}};
\fill (15.00000,-3.80000) circle (0.08cm);
\draw[->] (4.0000,-4.5500) -- (15.3000,-4.5500);
\draw (4.00000,-4.61000) -- (4.00000,-4.49000);
\node[anchor=north,font=\scriptsize] at (4.00000,-4.65000) {0};
\draw (5.91638,-4.61000) -- (5.91638,-4.49000);
\node[anchor=north,font=\scriptsize] at (5.91638,-4.65000) {5};
\draw (7.83275,-4.61000) -- (7.83275,-4.49000);
\node[anchor=north,font=\scriptsize] at (7.83275,-4.65000) {10};
\draw (9.74913,-4.61000) -- (9.74913,-4.49000);
\node[anchor=north,font=\scriptsize] at (9.74913,-4.65000) {15};
\draw (11.66551,-4.61000) -- (11.66551,-4.49000);
\node[anchor=north,font=\scriptsize] at (11.66551,-4.65000) {20};
\draw (13.58188,-4.61000) -- (13.58188,-4.49000);
\node[anchor=north,font=\scriptsize] at (13.58188,-4.65000) {25};
\node[anchor=north,font=\small] at (9.5000,-5.0000) {Running time (s)};

\end{tikzpicture}
}\\
\resizebox{0.7\textwidth}{!}{%
\begin{tikzpicture}
\node[anchor=south,font=\normalsize\bfseries] at (9.65,0.5) {Human chr19};
\node[anchor=east,font=\small] at (3.7500,0.0000) {\texttt{mincard}};
\draw (11.83839,0.00000) circle (0.13cm);
\fill (10.48661,0.00000) circle (0.08cm);
\fill (10.36328,0.00000) circle (0.08cm);
\fill (9.70812,0.00000) circle (0.08cm);
\fill (9.45121,0.00000) circle (0.08cm);
\fill (9.07762,0.00000) circle (0.08cm);
\fill (9.29760,0.00000) circle (0.08cm);
\fill (9.06181,0.00000) circle (0.08cm);
\fill (9.07818,0.00000) circle (0.08cm);
\fill (9.07623,0.00000) circle (0.08cm);
\fill (9.16547,0.00000) circle (0.08cm);
\fill (9.24353,0.00000) circle (0.08cm);
\fill (9.36805,0.00000) circle (0.08cm);
\fill (9.58453,0.00000) circle (0.08cm);
\fill (10.15550,0.00000) circle (0.08cm);
\fill (10.59392,0.00000) circle (0.08cm);
\node[inner sep=0pt,font=\large] at (11.01829,0.00000) {$\triangle$};
\node[anchor=east,font=\small] at (3.7500,-0.9500) {\texttt{mincard pc}};
\draw (12.91612,-0.95000) circle (0.13cm);
\fill (12.91782,-0.95000) circle (0.08cm);
\fill (12.86390,-0.95000) circle (0.08cm);
\fill (12.85845,-0.95000) circle (0.08cm);
\fill (12.97410,-0.95000) circle (0.08cm);
\fill (13.28008,-0.95000) circle (0.08cm);
\fill (12.83224,-0.95000) circle (0.08cm);
\fill (12.91061,-0.95000) circle (0.08cm);
\fill (12.67791,-0.95000) circle (0.08cm);
\fill (12.87915,-0.95000) circle (0.08cm);
\fill (12.92874,-0.95000) circle (0.08cm);
\fill (13.23070,-0.95000) circle (0.08cm);
\fill (13.63988,-0.95000) circle (0.08cm);
\fill (13.62113,-0.95000) circle (0.08cm);
\fill (13.82582,-0.95000) circle (0.08cm);
\fill (14.29422,-0.95000) circle (0.08cm);
\node[inner sep=0pt,font=\large] at (15.00000,-0.95000) {$\triangle$};
\node[anchor=east,font=\small] at (3.7500,-1.9000) {$\trivialsegmentationhorizontal$};
\fill (13.50594,-1.90000) circle (0.08cm);
\node[anchor=east,font=\small] at (3.7500,-2.8500) {$\trivialsegmentationvertical$ \texttt{pc}};
\fill (9.45121,-2.85000) circle (0.08cm);
\draw[->] (4.0000,-3.6000) -- (15.3000,-3.6000);
\draw (4.00000,-3.66000) -- (4.00000,-3.54000);
\node[anchor=north,font=\scriptsize] at (4.00000,-3.70000) {0};
\draw (6.57473,-3.66000) -- (6.57473,-3.54000);
\node[anchor=north,font=\scriptsize] at (6.57473,-3.70000) {500};
\draw (9.14945,-3.66000) -- (9.14945,-3.54000);
\node[anchor=north,font=\scriptsize] at (9.14945,-3.70000) {1000};
\draw (11.72418,-3.66000) -- (11.72418,-3.54000);
\node[anchor=north,font=\scriptsize] at (11.72418,-3.70000) {1500};
\draw (14.29890,-3.66000) -- (14.29890,-3.54000);
\node[anchor=north,font=\scriptsize] at (14.29890,-3.70000) {2000};
\node[anchor=north,font=\small] at (9.5000,-4.0500) {Running time (s)};

\end{tikzpicture}
}

\caption{Running time comparison on the largest SaRS-CoV-2 dataset, E.\ Coli dataset, and human chromosome 19 dataset. Here \texttt{mincard} is run in the gaps as symbol mode and the shown running times correspond to different values of parameter $U=4,8,\ldots$ with the smallest and largest marked with empty circle and $\triangle$, respectively\label{fig:runningtime-ecoli}.}
\end{figure}
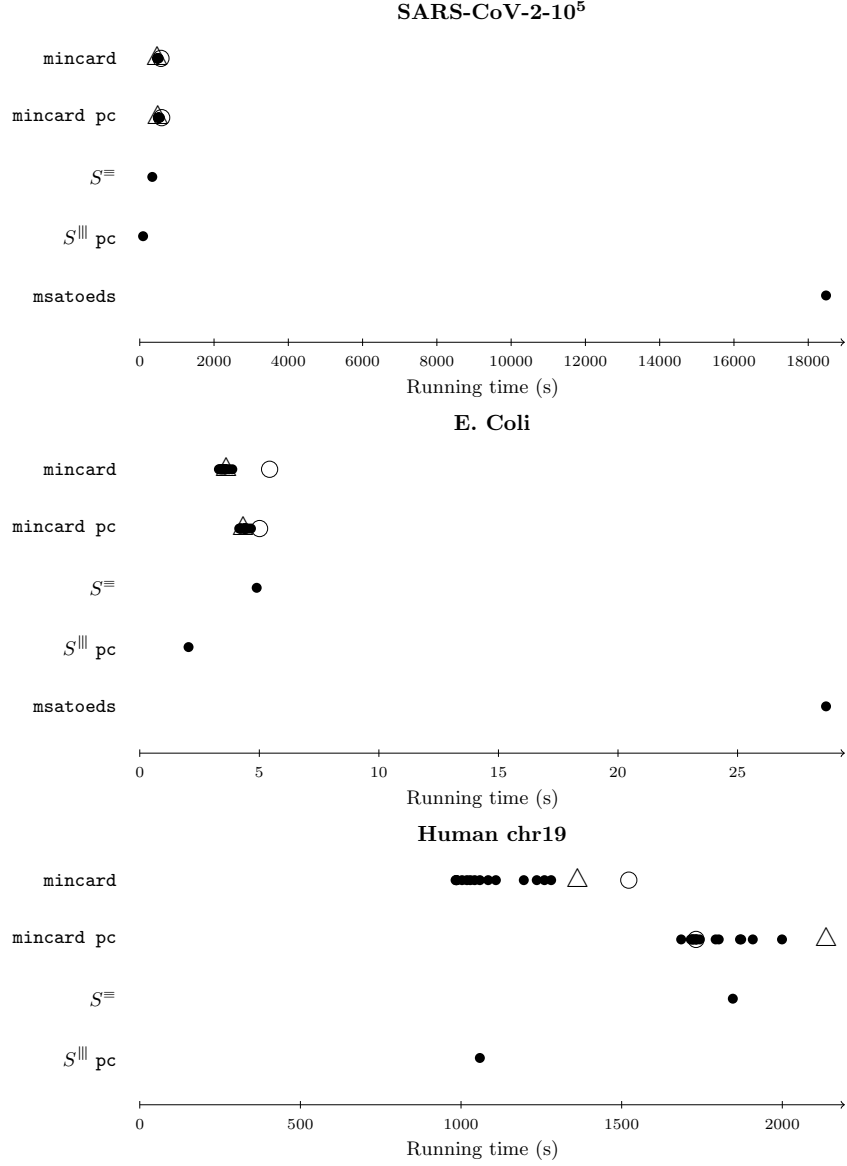

\section{Discussion}
We have introduced linear-time solutions to the minimum-cardinality (\minUcard, \Cref{prob:minUcard}) and the minimum-size  (\minLsize, \Cref{prob:minLsize}) segmentation problems, when the multiple sequence alignment in input is gapless and we impose a maximum or minimum segment length $U$ or $L$, respectively.
These constraints have an important effect on the features of the resulting EDS, and similar constraints are present in the literature of segmentation algorithms.
For the more realistic general case with gaps, our experiments indicate that treating gaps as symbols and removing them after segmentation does not degrade much the EDS of minimum cardinality in practice.
We speculate that similar results might hold for the related \minLsize\ problem. 
Finally, we remark that the solutions to \minUcard\ and \minLsize\ extend the literature on segmentation algorithms, and this line of research could contribute to versatile, practical, and scalable algorithm for the construction of pangenomes.

The following theoretical questions prevail:
\begin{description}
    \item[Gaps as symbols yield exact solutions?] 
    At least for height optimization, it seems possible to characterize the class of MSAs with gaps where the gaps-as-symbols approach yields an optimal solution, and not just an upper bound. Such characterization should be possible through forbidding local sub-optimal alignments like \texttt{A-} and  \texttt{-A}. This raises several natural questions: What is the exact form of the characterization? 
    How to efficiently detect if a given MSA is part of this class? Do all optimal alignments under some column-wise scoring function belong to this class? Are there efficient algorithms to convert an MSA into one in this class without sacrificing the alignment score?
    \item[Tailored algorithms with gaps.]
    The above approaches will not work for arbitrary MSAs and all quality measures (like size), so it is natural to ask if there are efficient algorithms to directly optimize the measures on general MSAs. Alternatively, there may be slight adjustments to the measures that are amenable to efficient algorithms: in different contexts, the notions of prefix-aware height and segment length (which is different than the string length after the removal of gaps) have been introduced or optimized \citep{Rizzo24efg}.
\end{description}

\section*{Code and data availability}

The source code and data to reproduce the experiments are publicly available at \url{https:github.com/algbio/eds}.

\section*{Supplementary material}

See \Cref{sect:running-times}.

\section*{Funding information}

This work was partially supported by the PANGAIA,  ALPACA, and TeamPerMed projects that received funding from the European Union’s Horizon 2020 research and innovation programme with two first under the Marie Skłodowska-Curie grant agreements No.\ 872539 and 956229, respectively, and the third under grant agreement No.\ 101060011. NP was partially supported by MUR PRIN 2022 YRB97K PINC.

\section*{Acknowledgements}

The authors wish to thank the Finnish Computing Competence Infrastructure (FCCI) for supporting this project with computational and data storage resources. The first prototype implementation of the software was generated using ChatGPT. ChatGPT was also used to generate code to transform part of the experiment logs into visualizations.


\newpage

\appendix

\section{EDS statistics and running times\label{sect:running-times}}

{\small\begin{longtable}{crcrrrrrrr}
    \caption{EDS Cardinalities and sizes obtained by \texttt{mincard} on the tested MSAs, upper bound on the segment length $U$, and optional use of the perfect segments (pc). The two tested modes are $\mathrm{OPT}$, using the standard meaningful left extension (\Cref{def:lextextensions}), and $\mathrm{OPT}_{\gap \in \Sigma}$, approximating the result with the gaps-as-symbols strategy (\Cref{obs:gaps-as-symbols}). Note that for the latter problem, the shown `card' value is the final EDS cardinality and not the DP value from \Cref{eq:recmeaningfulcardinality}.}~\label{tab:opt}\\
    \toprule
    \multirow{2}*{dataset} & \multirow{2}*{$U$} & \multirow{2}*{pc} & \multicolumn{2}{c}{$\mathrm{OPT}$} & \multicolumn{2}{c}{$\mathrm{OPT}_{\gap \in \Sigma}$} & \multicolumn{2}{c}{gaps-as-symbol effect} \\
    & & & card & size & card & size & card & size\\
    \midrule
    SARS-CoV-2-$10^2$ & 4 & & 8\,436 & 33\,306 & 8\,436 & 33\,306 & -- & -- \\
     & 8 & & 4\,648 & 36\,548 & 4\,648 & 36\,547 & -- & $-$0.00\% \\
     & 16 & & 2\,754 & 43\,001 & 2\,754 & 43\,001 & -- & -- \\
     & 32 & & 1\,800 & 55\,886 & 1\,800 & 55\,886 & -- & -- \\
     & 64 & & 1\,319 & 80\,347 & 1\,319 & 80\,347 & -- & -- \\
     & 128 & & 1\,061 & 126\,366 & 1\,061 & 126\,366 & -- & -- \\
     & 256 & & 921 & 204\,398 & 921 & 204\,398 & -- & -- \\
     & 512 & & 825 & 355\,528 & 825 & 355\,528 & -- & -- \\
     & 4 & \checkmark & 2\,345 & 31\,697 & 2\,345 & 31\,697 & -- & -- \\
     & 8 & \checkmark & 2\,073 & 33\,393 & 2\,073 & 33\,393 & -- & -- \\
     & 16 & \checkmark & 1\,792 & 38\,507 & 1\,792 & 38\,507 & -- & -- \\
     & 32 & \checkmark & 1\,512 & 50\,881 & 1\,512 & 50\,881 & -- & -- \\
     & 64 & \checkmark & 1\,261 & 76\,029 & 1\,261 & 76\,029 & -- & -- \\
     & 128 & \checkmark & 1\,057 & 123\,840 & 1\,057 & 123\,840 & -- & -- \\
     & 256 & \checkmark & 921 & 204\,398 & 921 & 204\,398 & -- & -- \\
     & 512 & \checkmark & 825 & 355\,528 & 825 & 355\,528 & -- & -- \\
     \midrule
    SARS-CoV-2-$10^3$ & 4 & & 11\,954 & 46\,832 & 11\,954 & 46\,837 & -- & $+$0.01\% \\
    & 8 & & 8\,143 & 63\,397 & 8\,143 & 63\,402 & -- & $+$0.01\% \\
    & 16 & & 6\,246 & 96\,231 & 6\,246 & 96\,230 & -- & $-$0.00\% \\
    & 32 & & 5\,299 & 159\,839 & 5\,299 & 159\,839 & -- & -- \\
    & 64 & & 4\,841 & 274\,541 & 4\,841 & 274\,541 & -- & -- \\
    & 128 & & 4\,621 & 469\,751 & 4\,621 & 469\,751 & -- & -- \\
    & 256 & & 4\,510 & 789\,414 & 4\,510 & 789\,414 & -- & -- \\
    & 512 & & 4\,460 & 1\,144\,693 & 4\,460 & 1\,144\,693 & -- & -- \\
    & 4 & \checkmark & 8\,913 & 43\,386 & 8\,913 & 43\,391 & -- & $+$0.01\% \\
    & 8 & \checkmark & 7\,353 & 59\,721 & 7\,354 & 59\,716 & $+$0.01\% & $-$0.01\% \\
    & 16 & \checkmark & 6\,125 & 93\,798 & 6\,125 & 93\,797 & -- & $-$0.00\% \\
    & 32 & \checkmark & 5\,291 & 159\,056 & 5\,291 & 159\,056 & -- & -- \\
    & 64 & \checkmark & 4\,841 & 274\,527 & 4\,841 & 274\,527 & -- & -- \\
    & 128 & \checkmark & 4\,621 & 469\,751 & 4\,621 & 469\,751 & -- & -- \\
    & 256 & \checkmark & 4\,510 & 789\,414 & 4\,510 & 789\,414 & -- & -- \\
    & 512 & \checkmark & 4\,460 & 1\,144\,693 & 4\,460 & 1\,144\,693 & -- & -- \\
    \midrule
    SARS-CoV-2-$10^4$ & 4 & & 24\,892 & 94\,207 & 24\,894 & 94\,218 & $+$0.01\% & $+$0.01\% \\
    & 8 & & 20\,710 & 155\,041 & 20\,714 & 155\,108 & $+$0.02\% & $+$0.04\% \\
    & 16 & & 18\,647 & 271\,363 & 18\,650 & 272\,029 & $+$0.02\% & $+$0.25\% \\
    & 32 & & 17\,681 & 480\,761 & 17\,683 & 481\,037 & $+$0.02\% & $+$0.06\% \\
    & 64 & & 17\,271 & 807\,248 & 17\,275 & 807\,579 & $+$0.02\% & $+$0.04\% \\
    & 128 & & 17\,141 & 1\,163\,355 & 17\,145 & 1\,162\,506 & $+$0.02\% & $-$0.07\% \\
    & 256 & & 17\,121 & 1\,312\,482 & 17\,125 & 1\,311\,633 & $+$0.02\% & $-$0.06\% \\
    & 512 & & 17\,119 & 1\,330\,925 & 17\,123 & 1\,330\,076 & $+$0.02\% & $-$0.06\% \\
    & 4 & \checkmark & 24\,577 & 93\,154 & 24\,579 & 93\,165 & $+$0.01\% & $+$0.01\% \\
    & 8 & \checkmark & 20\,692 & 154\,652 & 20\,696 & 154\,711 & $+$0.02\% & $+$0.04\% \\
    & 16 & \checkmark & 18\,647 & 271\,253 & 18\,650 & 271\,921 & $+$0.02\% & $+$0.25\% \\
    & 32 & \checkmark & 17\,681 & 480\,761 & 17\,683 & 481\,037 & $+$0.01\% & $+$0.06\% \\
    & 64 & \checkmark & 17\,271 & 807\,248 & 17\,275 & 807\,579 & $+$0.02\% & $+$0.04\% \\
    & 128 & \checkmark & 17\,141 & 1\,163\,355 & 17\,145 & 1\,162\,506 & $+$0.02\% & $-$0.07\% \\
    & 256 & \checkmark & 17\,121 & 1\,312\,482 & 17\,125 & 1\,311\,633 & $+$0.02\% & $-$0.06\% \\
    & 512 & \checkmark & 17\,119 & 1\,330\,925 & 17\,123 & 1\,330\,076 & $+$0.02\% & $-$0.06\% \\
    \midrule
    SARS-CoV-2-$10^5$ & 4 & & 55\,737 & 192\,852 & 55\,780 & 192\,385 & $+$0.08\% & $-$0.24\% \\
    & 8 & & 49\,845 & 335\,595 & 49\,889 & 334\,962 & $+$0.09\% & $-$0.19\% \\
    & 16 & & 46\,884 & 591\,668 & 46\,920 & 591\,518 & $+$0.08\% & $-$0.03\% \\
    & 32 & & 45\,603 & 983\,313 & 45\,634 & 981\,999 & $+$0.07\% & $-$0.13\% \\
    & 64 & & 45\,198 & 1\,395\,690 & 45\,230 & 1\,392\,327 & $+$0.07\% & $-$0.24\% \\
    & 128 & & 45\,145 & 1\,579\,930 & 45\,177 & 1\,576\,699 & $+$0.07\% & $-$0.20\% \\
    & 256 & & 45\,142 & 1\,618\,986 & 45\,174 & 1\,615\,755 & $+$0.07\% & $-$0.20\% \\
    & 512 & & 45\,142 & 1\,618\,986 & 45\,174 & 1\,615\,755 & $+$0.07\% & $-$0.20\% \\
    & 4 & \checkmark & 55\,731 & 192\,805 & 55\,774 & 192\,338 & $+$0.08\% & $-$0.24\% \\
    & 8 & \checkmark & 49\,845 & 335\,562 & 49\,889 & 334\,929 & $+$0.09\% & $-$0.19\% \\
    & 16 & \checkmark & 46\,884 & 591\,668 & 46\,920 & 591\,518 & $+$0.08\% & $-$0.03\% \\
    & 32 & \checkmark & 45\,603 & 983\,313 & 45\,634 & 981\,999 & $+$0.07\% & $-$0.13\% \\
    & 64 & \checkmark & 45\,198 & 1\,395\,690 & 45\,230 & 1\,392\,327 & $+$0.07\% & $-$0.24\% \\
    & 128 & \checkmark & 45\,145 & 1\,579\,930 & 45\,177 & 1\,576\,699 & $+$0.07\% & $-$0.20\% \\
    & 256 & \checkmark & 45\,142 & 1\,618\,986 & 45\,174 & 1\,615\,755 & $+$0.07\% & $-$0.20\% \\
    & 512 & \checkmark & 45\,142 & 1\,618\,986 & 45\,174 & 1\,615\,755 & $+$0.07\% & $-$0.20\% \\
    \midrule
    E.\ Coli (sim) & 4 & & 3\,639\,064 & 10\,979\,841 & 3\,640\,255 & 10\,984\,229 & $+$0.03\% & $+$0.04\% \\
    & 8 & & 2\,355\,853 & 13\,722\,719 & 2\,356\,661 & 13\,727\,963 & $+$0.03\% & $+$0.04\% \\
    & 16 & & 1\,726\,229 & 19\,147\,258 & 1\,726\,693 & 19\,157\,226 & $+$0.03\% & $+$0.05\% \\
    & 32 & & 1\,373\,078 & 30\,006\,581 & 1\,373\,358 & 30\,030\,869 & $+$0.02\% & $+$0.08\% \\
    & 64 & & 1\,024\,562 & 48\,505\,128 & 1\,024\,715 & 48\,521\,699 & $+$0.01\% & $+$0.03\% \\
    & 4 & \checkmark & 3\,284\,126 & 10\,250\,868 & 3\,285\,203 & 10\,254\,712 & $+$0.03\% & $+$0.04\% \\
    & 8 & \checkmark & 2\,271\,544 & 12\,814\,406 & 2\,272\,293 & 12\,819\,352 & $+$0.03\% & $+$0.04\% \\
    & 16 & \checkmark & 1\,716\,426 & 18\,577\,397 & 1\,716\,875 & 18\,587\,511 & $+$0.03\% & $+$0.05\% \\
    & 32 & \checkmark & 1\,372\,418 & 29\,800\,277 & 1\,372\,700 & 29\,824\,290 & $+$0.02\% & $+$0.08\% \\
    & 64 & \checkmark & 1\,024\,460 & 48\,236\,460 & 1\,024\,613 & 48\,252\,330 & $+$0.01\% & $+$0.03\% \\
    \midrule
    Human chr19 & 4 & & 15\,850\,571 & 62\,605\,930 & 15\,850\,710 & 62\,606\,319 & $+$0.00\% & $+$0.00\% \\
    & 8 & & 8\,369\,866 & 66\,084\,412 & 8\,370\,005 & 66\,085\,534 & $+$0.00\% & $+$0.00\% \\
    & 16 & & 4\,631\,816 & 73\,124\,306 & 4\,631\,958 & 73\,126\,485 & $+$0.00\% & $+$0.00\% \\
    & 4 & \checkmark & 2\,641\,842 & 60\,608\,317 & 2\,641\,914 & 60\,608\,709 & $+$0.00\% & $+$0.00\% \\
    & 8 & \checkmark & 2\,427\,667 & 61\,835\,634 & 2\,427\,787 & 61\,836\,508 & $+$0.00\% & $+$0.00\% \\
    & 16 & \checkmark & 2\,181\,927 & 65\,990\,892 & 2\,182\,038 & 65\,992\,345 & $+$0.01\% & $+$0.00\% \\
    \bottomrule
\end{longtable}}

The linear-time dependency on $U$ in the naive computation of the meaningful left extensions can be clearly seen in \Cref{fig:scalability}.

\begin{figure}
    \centering
    \includegraphics[]{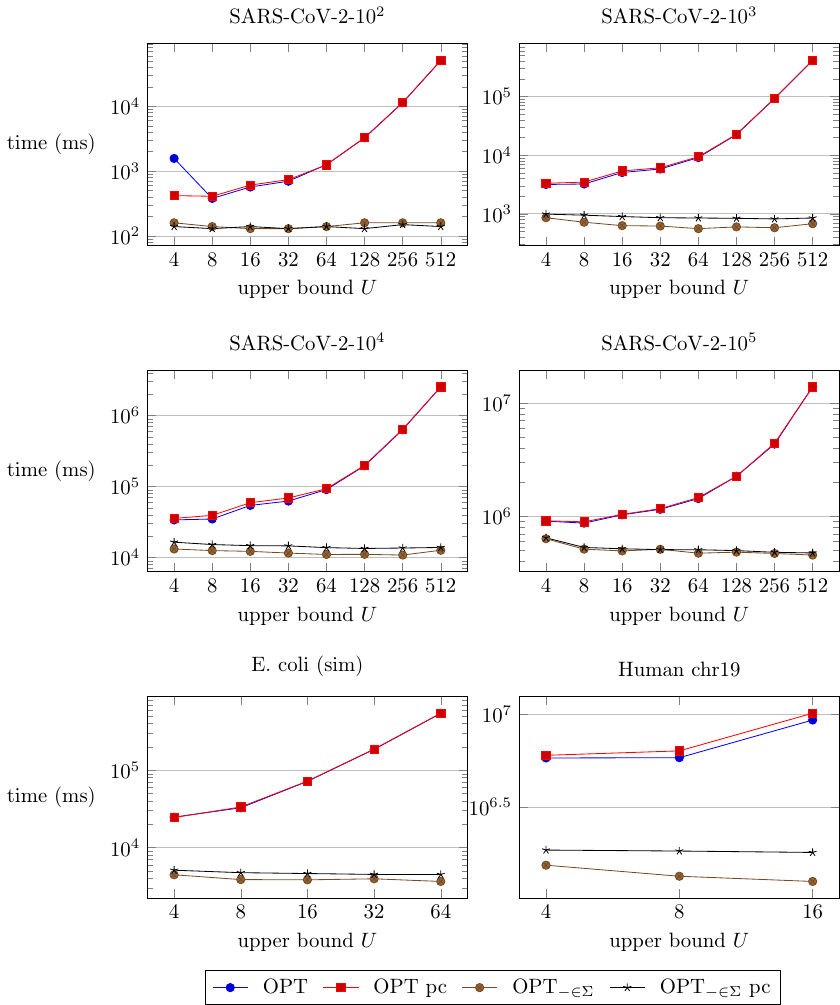}
    \caption{Log-log plot of the running times of the different \texttt{mincard} modes against the upper bound $U$ on the segment length, where $\mathrm{OPT}$ considers the standard meaningful left extensions (\Cref{def:lextextensions}), $\mathrm{OPT}_{\gap \in \Sigma}$ approximates the result with the gaps-as-symbols strategy (\Cref{obs:gaps-as-symbols}) and uses the Positional Burrows--Wheeler transform, and modifier pc considers the perfect segments of any length.} 
    \label{fig:scalability}
\end{figure}

\end{document}